\documentclass{article}
\usepackage{fullpage}
\usepackage{graphicx,amssymb,amsmath, amsthm}
\usepackage{commath, enumitem}
\usepackage{hyperref}
\usepackage[capitalise,noabbrev]{cleveref}

\usepackage{subcaption}
\usepackage{enumitem}

\newtheorem{theorem}{Theorem}[section]
\newtheorem{cor}[theorem]{Corollary}
\newtheorem{lemma}[theorem]{Lemma}
\newtheorem{obs}[theorem]{Observation}

\crefname{subsection}{subsection}{subsections}
\crefname{obs}{Observation}{Observations}

\graphicspath{{cccg/graphics/}}

\title{Carving Polytopes with Saws in 3D}
\author{
Eliot W. Robson\thanks{Department of Computer Science,
		University of Illinois Urbana-Champaign, {\tt erobson2@illinois.edu}}
\and
Jack Spalding-Jamieson\thanks{David R. Cheriton School of Computer Science, University of Waterloo, {\tt jacksj@uwaterloo.ca}}
\and Da Wei Zheng\thanks{Department of Computer Science,
		University of Illinois Urbana-Champaign, {\tt dwzheng2@illinois.edu}}
}

\newcommand{\Int}[1]{\operatorname{int} \del{#1}}
\newcommand{\CH}[1]{\operatorname{CH} \del{#1}}

\newcommand{\R}{\mathbb{R}}

\newcommand{\Halfplanes}{\mathcal{H}}
\newcommand{\Rays}{\ensuremath{\mathcal{R}}}
\newcommand{\whp}{w.h.p.}

\newcommand{\eps}{\varepsilon}

\begin{document}
\thispagestyle{empty}
\maketitle

\begin{abstract}
We investigate the problem of 
\emph{carving} an $n$-face triangulated three-dimensional polytope
using a tool to make cuts modelled by either a half-plane or sweeps from an infinite ray.
In the case of half-planes cuts, we present a deterministic algorithm running in $O(n^2)$ time and a randomized algorithm running in $O(n^{3/2+\varepsilon})$ expected time for any $\varepsilon>0$.
In the case of cuts defined by sweeps of infinite rays, we present an algorithm running in $O(n^5)$ time.
\end{abstract}

\section{Introduction}
\label{sec:intro}

Stone carving is one of the earliest known representational works of art, and has
been known to predate even the earliest human civilization. This is the practice of taking a single solid piece of material and removing pieces until achieving a desired final shape.
To ensure durability of the finished product, the base
material is often very durable and can be difficult to carve with tools.
As a result, it is desirable to minimize the amount of work that must be done to achieve the final carving, and is useful to be able to determine what kind of objects can be carved out with the tools being used.

\subsection{2D Cutting}
The two-dimensional case of cutting material
was first studied by Overmar and Welzl \cite{OvermarsW85}. In this work, the
authors modeled cuts as lines in the plane, giving algorithms for computing
the cheapest sequence of cuts in special cases.
This was generalized by
Demaine, Demaine, and Kaplan \cite{DemaineDK01} to the case of cutting with line segments in the plane where they gave an algorithm with $2.5$ approximation factor. 
This approximation factor was later improved by Dumitrescu and Hasan \cite{DumitrescuH13}.
In 2009, Bereg, Daescu, and Jiang \cite{BeregDJ09} presented a PTAS for the  
problem of minimum length when cutting convex $n$-gons out of convex $m$-gons with straight line cuts.

An analogous problem has been studied for ray cuts \cite{DaescuL06, Tan05}, where the cutting object is a ray instead of a line or line
segment. These results focus on carving both convex and simple polygons,
and minimizing the length of cuts.
In \cref{sec:ray-cuts}, we study a more general version of the decision variant of this problem,
where we find maximum ray-carveable regions not crossing a set of disjoint polygons.

\begin{figure}
    \centering
    \includegraphics[scale=0.2]{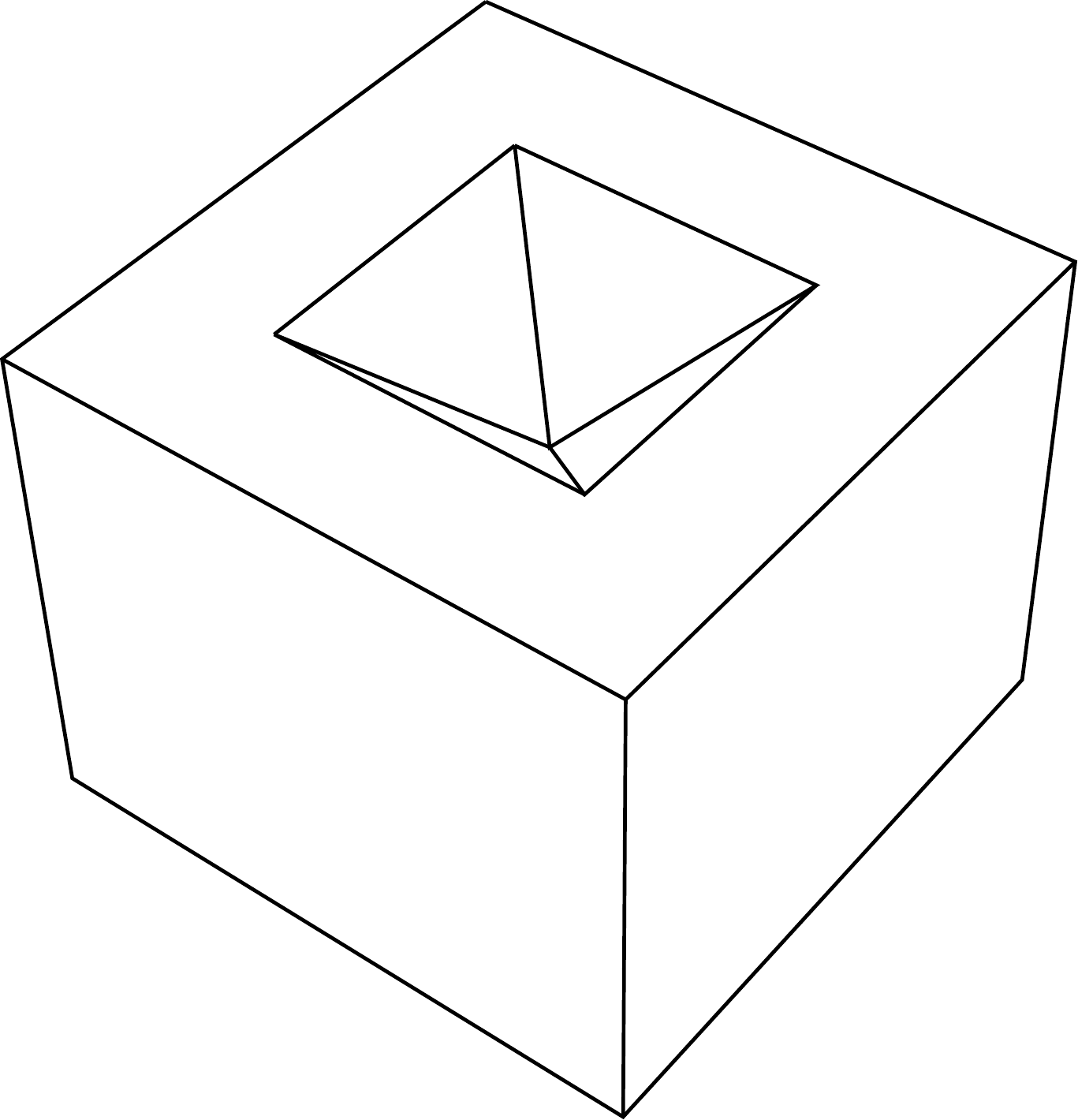}
    \qquad
    \includegraphics[scale=0.23]{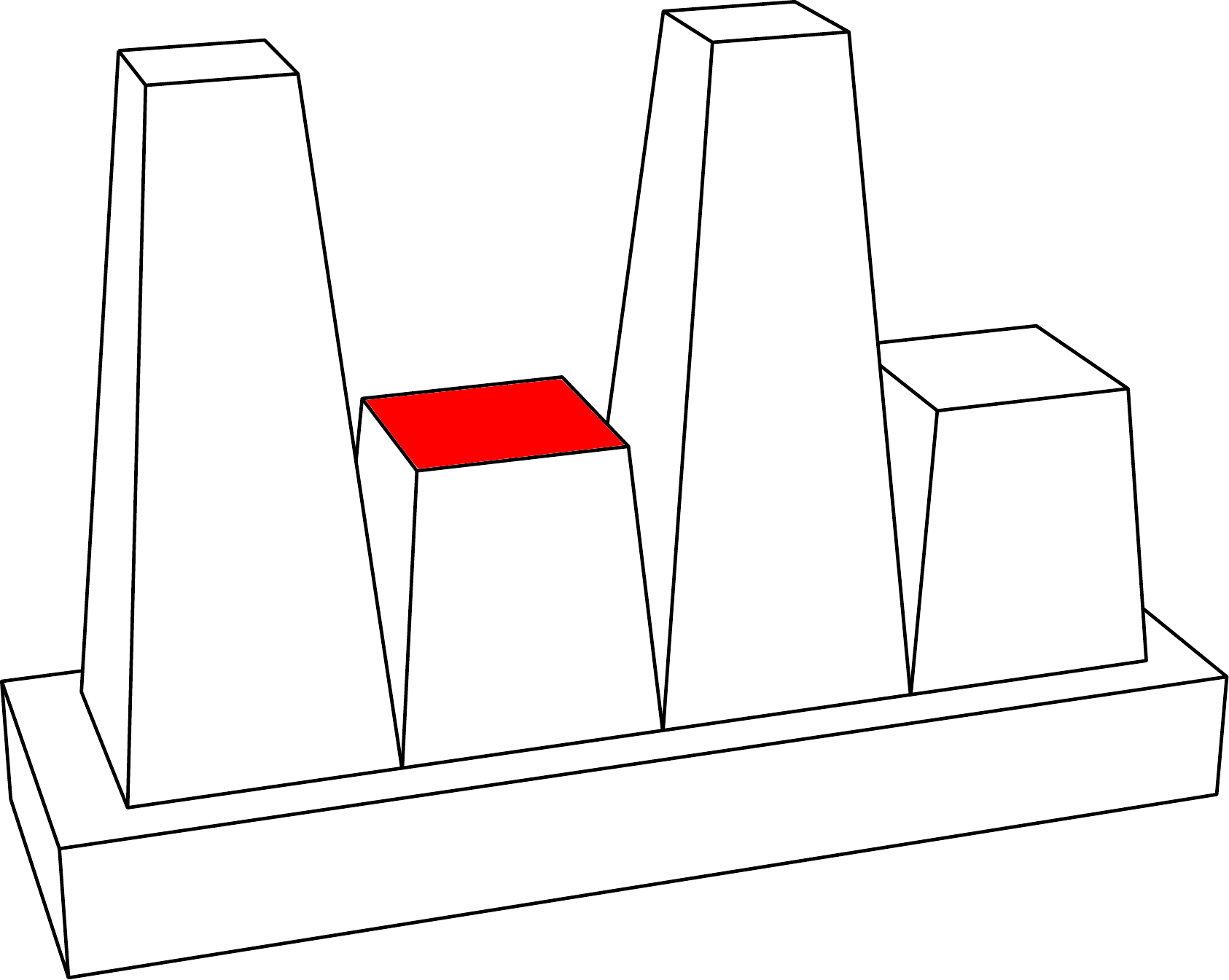}
    
    \caption{\textbf{(Left.)} A polytope containing a cavity that can be carved with ray sweeps but not half-planes. \textbf{(Right.)} A polytope that can be carved with rays, but with exactly one face that cannot be carved by half-plane cuts.}
    \label{fig:ray-only-simple}
\end{figure}

\begin{figure}
    \centering
    \label{fig:ray-only}
\end{figure}

\subsection{3D Cutting}

Surprisingly, we found very little work on 3D generalizations of these problems, the only ones being guillotine cuts (cuts that go all the way through) used to cut a convex polygon out of a sphere \cite{AhmedHI11}, and work on cutting styrofoam with hot wire \cite{JaromczykK03}.

We explore carving three-dimensional shapes with models of two different classes of cuts:
\begin{itemize}[noitemsep,topsep=0pt,parsep=0pt,partopsep=0pt]
\item Straight-cuts with tools like circular saws or table saws,
modelled using half-plane cuts (see 
    \cref{fig:genus-1-polytope}).
\item Tools with the ability to pierce up to a specific depth such as waterjets or laser cutters, modelled using ray sweeps (see
    \cref{fig:ray-only-simple}%
).
\end{itemize}
Note that straight-cuts can also be performed with a wide variety of other tools,
such as band saws, the long edge of chainsaws, or even jigsaws.
Some tools also allow for additional types of cuts that we do not model, such as using the tip of a chainsaw, or rounded cuts using a band saw.
However, straight-cuts can be performed with every sufficiently large saw,
and moreover they are the only useful type of cut for carving polytopes exactly, since polytopes do not have rounded edges.

For simplicity, we only consider 3D polytopes $P$ with $n$ vertices.
We make no assumptions of general position in this work,
although we only consider polytopes without self-intersections
or degenerate faces.
In other words,
we only consider polytopes that are uniquely defined by their connected interiors.

We will assume we are always given a polytope $P$ along with a triangulation of each face.
If not, we may triangulate each face in linear time \cite{chazelle1991triangulating} with no impact on the final running time. We henceforth refer to the resulting triangles as the facial triangles.
We will use $n$ denote the complexity of $P$, i.e. the number of facial triangles.
Note that this is asymptotically proportional to the number of vertices of $P$.
\begin{figure}
    \centering
    \includegraphics[scale=0.205]{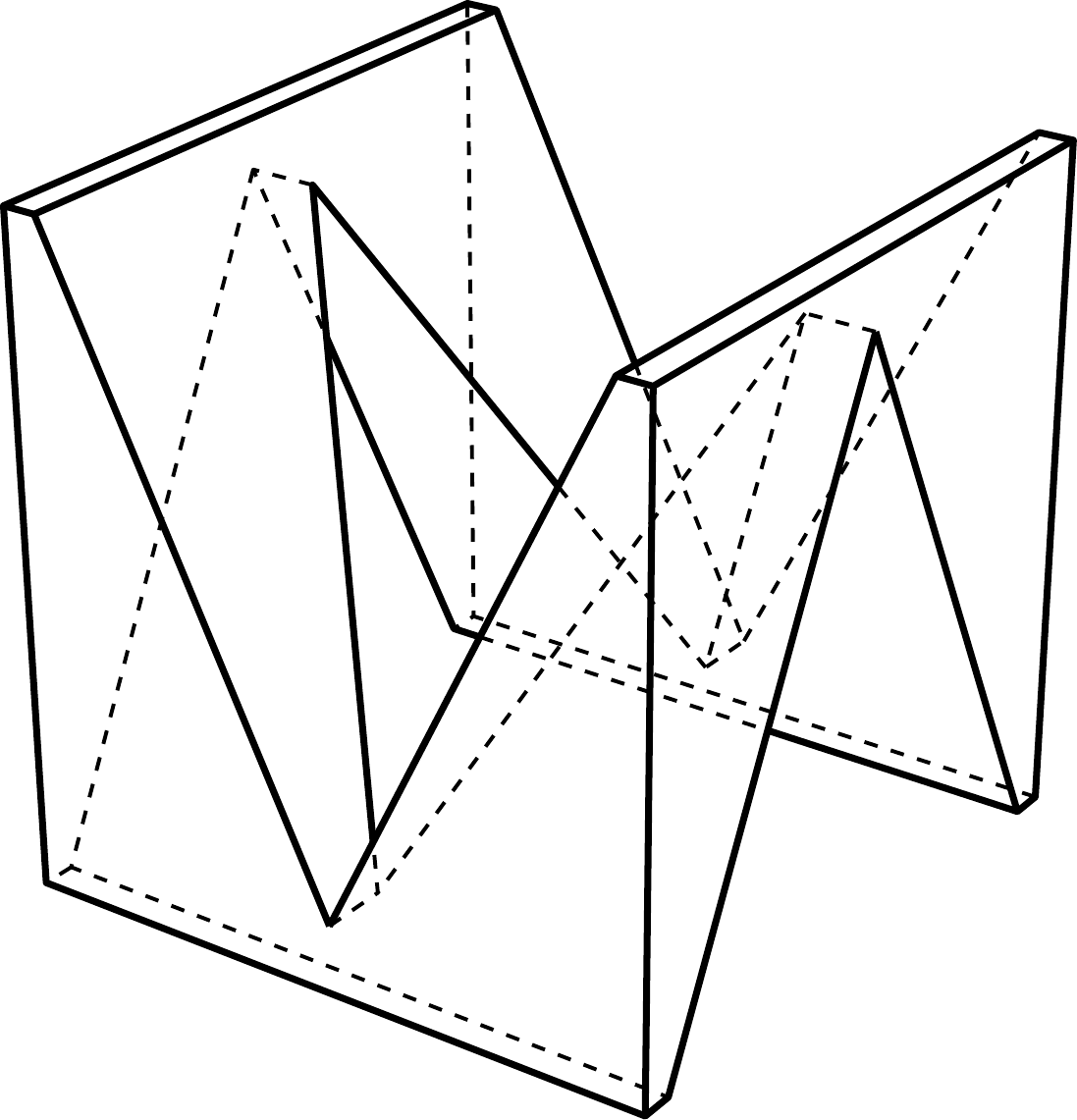}
    \qquad
    \includegraphics[scale=0.20]{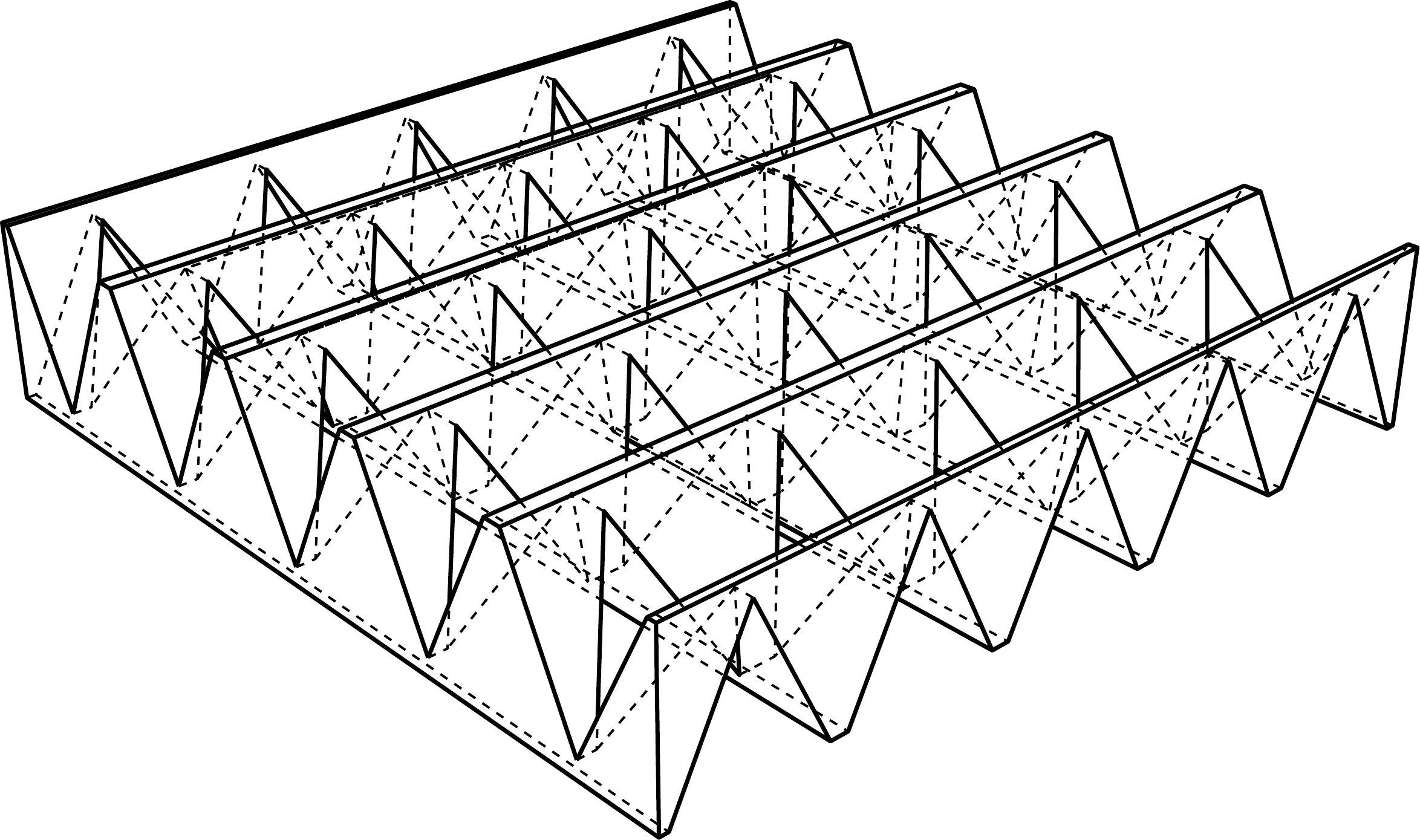}
    \caption{Non-convex polytopes with holes that can be carved using half-planes.}
    \label{fig:genus-1-and-n}
    \label{fig:genus-1-polytope}
    \label{fig:genus-n-polytope}
\end{figure}

\subsection{Our Contributions}
In our paper, we consider a polytope $P$
and determine if there exists a finite set of cuts $K$
such that for any
set $C\supset P$,
the ``carved'' set $C\setminus K$
has a connected component (i.e., maximal connected open subset) equal to $\Int{P}$ (the \emph{interior} of $P$).
Essentially, the excess connected components (``material'') can be ``removed''
to leave only the intended shape. We also wish this to be independent of $C$ (i.e.\ $C$ is not given as input to the algorithm) so that the produced set of cuts can be used to carve an object, regardless of the initial material.

In \cref{sec:half-plane} we consider 
the set of cuts to be half-planes, which can include complicated polytopes with holes as in \cref{fig:genus-1-and-n}.
We are able to characterize the shapes that can or cannot be carved,
and use this to derive a simple deterministic algorithm that runs in $O(n^2)$ time.
Furthermore, we also discuss a randomized algorithm that runs in $O(n^{3/2+\eps})$ time for any $\eps > 0$.

In \cref{sec:ray-cuts},
we model ray cuts as a finite set of \emph{ray sweeps} that consist of bounded continuous movement of a ray on a plane. 
Sweeps are meant to model cuts one could perform or program a machine to perform,
rather than simply being the set of shapes that are excluded by a infinite union of rays.
This model allows for a 
more general class of cuts than half-plane cuts,
as in \Cref{fig:ray-only-simple} and \cref{fig:ray-only}.
We are able to characterize the shapes that can be carved with ray sweeps,
and in this case
we present an algorithm that runs in $O(n^5)$ time.


\section{Half-Plane Cuts}
\label{sec:half-plane}
In this section, we study half-plane cuts which model cuts that can be performed with tools such as circular saws or table saws.
Formally, a half-plane is 
defined by a plane in $\R^3$ and the region on one side of a line contained in that plane. 
A polytope $P$
is \emph{carveable} if there is a 
set of half-plane cuts $\Halfplanes{}$,
so that for any $P$-containing set $C\supset P$,
the set $C\setminus \del{\cap_{H \in \Halfplanes{}} H }$
has a connected component (i.e., maximal connected open subset) equal to $\Int{P}$ (the \emph{interior} of $P$).

This model is equivalent to that of 
an
open question posed by Demaine et al.~\cite[$7$th open problem]{DemaineDK01}.
They ask if there exists an algorithm 
to cut three dimensional polyhedra
using an infinitely long rectangle that can only slice straight.

\subsection{Characterization of Carveable Polytopes} 

For a set $S$, we denote the convex hull of $S$ by $\CH{S}$.
We give a complete characterization of the polytopes $P$ which can be carved in this
model.

\begin{theorem}
\label{thm:half-plane-cutting-equivalence}
For a triangulated polytope $P$,
the following are equivalent:
\begin{enumerate}[noitemsep,topsep=0pt,parsep=0pt,partopsep=0pt]
    \item[(a)] $P$ is carveable.

    \item[(b)] For each facial triangle $T$ in $P$, there is a single half-plane $H_T$
    containing $T$ such that $H_T \cap \Int{P} = \varnothing$.

    \item[(c)] For each facial triangle $T$ in $P$,
        let $L_T$ denote the plane containing $T$.
        There is a single half-plane $H$ containing $T$
        and not containing  $\CH{\Int{P} \cap L_T}$
        whose boundary line passes through a vertex of 
        $\CH{\Int{P} \cap L_T}$,
        and also passes through a vertex $v$ of $T$.
\end{enumerate}
\end{theorem}

Essentially, the transformation from $(a)$ to $(b)$
indicates that it suffices to consider one cut per facial triangle,
and the transformation from $(b)$ to $(c)$ indicates that
it suffices to consider only a limited set of potential cuts for each facial triangle.

\begin{proof}
\hfill
\begin{enumerate}[align=left, leftmargin=*,noitemsep,topsep=0pt,parsep=0pt,partopsep=0pt]
    \item[$(a)\Rightarrow(b)$]
    Assume that $P$ can be carved by a set of half-planes $W$.
    For a facial triangle $T$ lying on the plane $L_T$, 
    let $W_T \subset W$ be the subset of half-planes on $L_T$.
    Then, by definition,
    $T \subset \bigcup_{w\in W_T} w$.
    We claim that there exists a half-plane $H\supset T$
    such that $H\cap\Int{P} = \varnothing$.
    Then
    $\bigcup_{w\in W_T} w$
    is the complement of an open convex set on $L_T$.
    Denote this open convex set on $L_T$ as $S$.
    Note that $S$ contains $\Int{P}\cap L_T$,
    since each $w\in W_T$ does not intersect $\Int{P}$.
    As $S$ and $T$ are disjoint convex sets, 
    there exists a separating line $l$
    between $S$ and $T$ on $L_T$,
    which induces a half-plane $H$
    with boundary $l$ containing $T$
    satisfying $H\cap \Int{P} = \varnothing$,
    as desired.
    
    \item[$(b)\Rightarrow(a)$]
    Take all $f$ cuts of the form $H_T$.
    This divides the exterior of $P$ in any $P$-containing set $C$,
    and does not intersect $P$, leaving it as a connected component.
    
    \item[$(b)\Rightarrow(c)$]
    Let $T$ be a facial triangle of $P$,
    and
    let $L_T$ denote the plane containing $T$.
    Assume that $L_T\cap\Int{P}\neq\varnothing$.
    Let $H_T$ be a half-plane
    on $L_T$ containing $T$ and not intersecting $P$, 
    i.e. $H_T\cap T=T$ and $H_T\cap\Int{P}=\varnothing$.
    A half-plane $H_T'\supset H_T$ can be found
    such that $H_T'$ touches the boundary of $\CH{\Int{P}\cap L_T}$ by translating $H_T$.
    Another half-plane $H_T''$ can be found
    by rotating $H_T'$ around the boundary of $\CH{\Int{P}\cap L_T}$
    until the result touches a vertex of $T$.
    $H_T''$ is then tangent to $\CH{\Int{P}\cap L_T}$,
    and some vertex $v$ of $T$ is on its boundary.
    
    \item[$(c)\Rightarrow(b)$] Trivial.
    \qedhere
\end{enumerate}
\end{proof}

We now present observations about facial triangles $T$ and $\CH{\Int{P} \cap L_T}$, where $L_T$ is the plane on which $T$ lies. The first is stated in somewhat greater generality. 

\begin{obs}\label{obs:edges}
Let $L$ be a plane,
and let $E$ be the set of line segments defining edges of $P$ that cross $L$.
That is, each line segment $e\in E$ has two endpoints which are strictly separated by $L$ (i.e. the endpoints of $e$ lie on opposite sides of $L$).
Then, $\CH{\Int{P} \cap L} = \CH{E \cap L}$.
\end{obs}
This is because the boundary of $\Int{P} \cap L$ consists of vertices that defined by the intersections of 
$E$ with $L$.

We make an additional observation about polytopes that are not carveable.
\begin{obs} \label{obs:triple_edges}
Let $v$ be a vertex of $P$ on facial triangle $T$ contained in the plane $L_T$.
If there exists three edges $e_1, e_2, e_3$ of $P$ such that $v\in \CH{\{e_1, e_2, e_3\} \cap L_T}$,
then the polytope $P$ is not carveable.
\end{obs}
This directly follows from \cref{thm:half-plane-cutting-equivalence}(c), as $\CH{\{e_1, e_2, e_3\} \cap L_T} \subseteq \CH{\Int{P} \cap L_T}$.

\subsection{Quadratic Time Decision Algorithm}
\label{subsec:quadratic}

Using the characterization from \cref{thm:half-plane-cutting-equivalence}(c) and the observations, we show a simple quadratic time algorithm exists for determining if a polytope is carveable.

\begin{figure}
    \centering
    \includegraphics[scale=0.8]{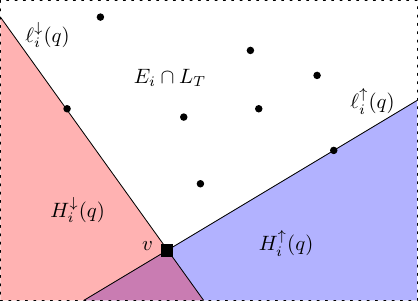}
    \caption{The tangents of $v$ and the halfplanes $H^\uparrow_i(v)$ and $H^\downarrow_i(v)$. The figure is drawn on the plane $L_T$. }
    \label{fig:tangents}
\end{figure}

\begin{theorem}\label{thm:quad-half-plane}
Given a triangulated polytope $P$ with $f$ facial triangles, 
there is an algorithm
to determine if $P$ can be carved by
half-planes
that runs in $O(n^2)$ time.
If the answer is in the affirmative, the algorithm
also outputs a set of $n$ half-planes
that carve $P$.
\end{theorem}

\begin{proof}
Let $T$ be a facial triangle,
and let $L_T$ be the plane containing $T$.
By \Cref{obs:edges},
we can compute the set of edges $E$ of $P$ that cross $L$
to compute $\CH{\Int{P} \cap L_T}$ in $O(n\log n)$ total time.
If the region is empty,
then any half-plane containing $T$ can be output.
Otherwise,
by characterization $(c)$
from \cref{thm:half-plane-cutting-equivalence},
we must determine if there exists a half-plane $H$
containing $T$
and tangent to
$\CH{\Int{P}\cap L_T}$,
such that the boundary line of $H$
passes through a vertex $v$ of $T$.
Note that each vertex $v$ of $T$ induces up to two lines 
going through $v$ and tangent to $\CH{\Int{P}\cap L_T}$.
It suffices to consider each of them.
It is possible to compute tangents in $O(\log n)$ time 
per vertex once we have explicitly computed the convex hull.

However, we can shave the log factor (pun intended), and
do this in $O(n)$ time per facial triangle with a different algorithm.
The tangent to $\CH{\Int{P}\cap L_T}$ on $L_T$ through a vertex $v$ 
is a line through a vertex $v$ and a vertex of $X = E\cap L_T$.
Let $l_{x,v}$ be the line
passing through a
point $x$ in $X$ and a vertex $v$ of $T$.
There are at most $O(n)$ such lines since $|E| = O(n)$.
If $l_{x,v}$ is a separating line w.r.t.~the plane $L_T$
between $X\setminus\{x\}$ and the other two vertices of $T$,
then the half-plane $H$
containing $T$ with boundary $l_{x,v}$
satisfies all of our conditions.
Otherwise, if no such line exists,
then no half-plane satisfying our conditions exists,
and $P$ cannot be carved using half-planes
by
\cref{thm:half-plane-cutting-equivalence}.

It is not immediately clear how we can 
efficiently check each of these lines.
However, observe that
only the ``extreme'' lines $l_{x,v}$ for any specific $v$
can be candidate separating lines.
This can be accomplished by performing what is essentially an iteration of the gift wrapping algorithm
for convex hulls~\cite[Section 1.1]{BergCKO08}.
Maintain a ``current'' line $l_{x',v}$.
Iterate through all the points $x\in X$.
If $x$ lies clockwise (resp.~counter-clockwise) relative to the ray from $v$ to $x'$, set $x'\leftarrow x$.

If $v$ lies outside of $\CH{X}$, this procedure is guaranteed to find the tangents to $\CH{X}$ passing through $v$ as candidate lines, and we can check in linear time whether the candidate lines separate $T$ and $X$.
If $v$ lies within $\CH{X}$, then no such line exists, so we conclude that the polytope $P$ is not carveable.

This algorithm shows how to find a 
half-plane for each facial triangle in $O(n)$ time or deduce one doesn't exist.
The whole algorithm runs in $O(n^2)$ time,
since there are $O(n)$ facial triangles.
\end{proof}

\begin{figure}
    \centering
    \includegraphics[scale=0.18]{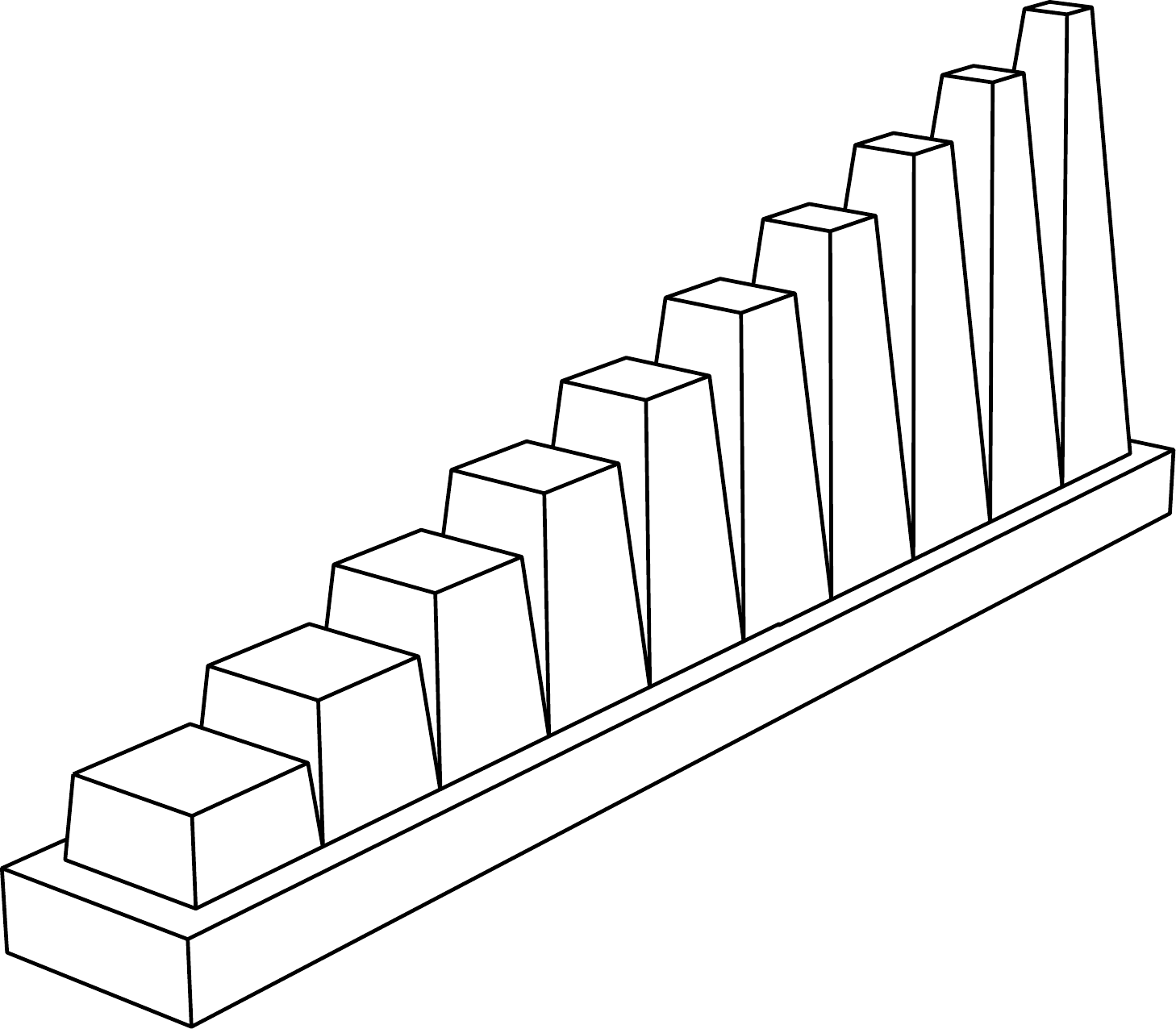}
    \caption{A construction for family of polytopes carveable using half-planes
    (to increase the size, add more pillars).
    For a member of this family with $n$ facial triangles, the algorithm described in
    \cref{thm:quad-half-plane} takes $\Theta(n^2)$ time.
    }
    \label{fig:slow-example}
\end{figure}

\Cref{fig:slow-example} demonstrates that there exists a family of polytopes
that can be carved using half-planes,
for which the runtime of the algorithm described in
\cref{thm:quad-half-plane} is $\Theta(n^2)$.
In the next section we show a faster algorithm for deciding if a polytope is carveable.

\subsection{Subquadratic Time Decision Algorithm}
\label{subsec:randomized-subquadratic}

At a high level, our quadratic time algorithm iterates through a set $\mathcal{Q}$ of queries consisting of pairs $q = (v, T)$,
where $v$ is a vertex on a facial triangle $T$ in the polytope.
Each query $q=(v, T)$ asks for tangents of $\CH{E \cap L_T}$ that go through $v$,
lying on the plane $L_T$ that contains $T$.
Since we have $|\mathcal{Q}| = O(n)$ queries and the polygon we are querying can be different for different planes $L_T$ and can also itself have size $O(n)$, this seems to require quadratic time.
However, we note that the polygons are related: They all come from the original polytope $P$. 
Surprisingly, we show that we can solve this problem faster by answering all of these queries simultaneously.

We present a reduction from the decision problem of whether a polytope $P$ is carveable using half-planes to the problem of detecting intersections between half-planes and line segments in $\R^3$.   
The reductions solves each query of $\mathcal{Q}$ on an increasingly large random subset of the edges $E$ that define the polytope $P$, by computing the violations between half-planes induced by the previous tangents, and using the randomized analysis of Clarkson--Shor~\cite{ClarksonS89} to bound the number of violations.
The violations are found by solving a problem involving intersections between half-planes and line segments.

\begin{lemma}
Let $T(n, k)$ be the time complexity of an algorithm for detecting intersection between $n$ half-planes and $n$ line segments in $\R^3$ with at most $k$ intersections.
Then, for any parameter $2\le r < n$, 
there is a randomized algorithm that decides if a polytope $P$ is carveable by half-planes that runs in expected time $O((\log n/\log r) \cdot (T(O(n), O(nr\log n)) + nr\log n))$ with high probability.
\end{lemma}

\begin{proof}
We use a bottom-up sampling approach. 
Let $E_0$ be the edges (i.e. line segments not including end points) 
defining the polytope $P$.
We choose $E_0\supset E_1\supset E_2\supset\cdots\supset E_k = \varnothing$ to be a series of uniformly random samples of the edges:
To get $E_{i+1}$ from $E_i$,
we take each edge $e\in E_i$ with probability $\frac1r$ for some parameter $r$.
We stop at the first value $k$ such that $E_k = \varnothing$. Note that with high probability (\whp{}) $k = \Theta(\log n/\log r)$, as $\abs{E_0} = O(n)$.

Throughout the algorithm, for each query $q = (v, T)$ pair (with $T$ contained in a plane $L_T$),
we maintain two half-planes $H^\uparrow_i(q), H^\downarrow_i(q) \subset L_T \setminus \CH{E_i \cap L_T}$. Their boundaries are, respectively, the upper and lower tangents of $\CH{E_i \cap L_T}$ passing through $v$.
Note that it is possible that $\CH{E_i \cap L_T}$ is empty (i.e.\ when $i=k$), so in this case, we let $H^\uparrow_i(q) = H^\downarrow_i(q) = L_T$. See \Cref{fig:tangents} for an illustration.
Let $\mathcal{H}_i = \{H^\uparrow_i(q) \mid q\in \mathcal{Q}\} \cup \{ H^\downarrow_i(q)\mid q\in \mathcal{Q}\}$.

To compute $\Halfplanes{}_{i}$ from $\Halfplanes{}_{i+1}$, we create an instance of intersection detection between the line segments $E_{i}\setminus E_{i+1}$ and the half-planes $\Halfplanes{}_{i+1}$. 

We analyze what intersections can occur for the half-planes defined by a query $q = (v,T)$.
We observe that on the plane $L_T$, since $E_{i+1}$ is a $1/r$ sample of $E_{i}$, $L_T\cap E_{i+1}$ is also a $1/r$ sample of $L_T\cap E_{i}$. 
By a standard analysis of Clarkson and Shor \cite{ClarksonS89}, this implies that the number of points of $L_T\cup E_{i}$ that lie within $H_{i+1}^\uparrow(q)$ and $H_{i+1}^\downarrow(q)$ is at most $O(r\log n)$ \whp{}
Note that this is exactly the intersections between $E_{i}\setminus E_{i+1}$ and $H_{i+1}^\uparrow(q)$ and $H_{i+1}^\downarrow(q)$, and thus the total number of intersections between $E_i\setminus E_{i+1}$ and $\Halfplanes{}_{i+1}$ is $O(nr\log n)$.

For a query $q=(v,T)$, three types of events may occur:
\begin{enumerate}[noitemsep,topsep=0pt,parsep=0pt,partopsep=0pt]
    \item \textbf{$E_{i+1}\cap L_T$ was empty and $E_{i}\cap L_T$ is non-empty.}
    In this case, we can inspect the points of $E_{i}\cap L_T$ (of which there at most $O(r\log n)$ \whp{}) and either compute $H^\uparrow_i(q)$ and $H^\downarrow_i(q)$ or deduce that $P$ is not carveable in linear time in the size of $E_{i}\cap L_T$ as in \cref{thm:quad-half-plane}.
    \item 
    \textbf{$E_{i}\cap L_T$ contains a point that lies in both $H^\uparrow_{i+1}$ and $H^\downarrow_{i+1}$.}
    Let $e$ denote the edge that induces this point.
    Let $e^\uparrow$ and $e^\downarrow$ be the edges of $E_{i+1}$ that defined $H^\uparrow_{i+1}$ and $H^\downarrow_{i+1}$ respectively.
    Observe that $e, e^\uparrow, e^\downarrow$ are the triple of edges that certify that $P$ is not carvable by \Cref{obs:triple_edges}.
    \item 
    \textbf{The points of $E_{i}\cap L_T$ either lie in $H^\uparrow_{i+1}(q)$ or $H^\downarrow_{i+1}(q)$.} 
    In this case, we can compute $H^\uparrow_{i}(q)$ and $H^\downarrow_{i}(q)$ from 
    $E_{i}\cap L_T$ in linear time following a similar procedure to find the most extreme halfplane as in \cref{thm:quad-half-plane}.
\end{enumerate}

In all cases, either we conclude that the polytope $P$ is not carveable or we compute $\mathcal{H}_{i}$. 
If we have computed $\mathcal{H}_0$, we have extreme half-planes for every vertex of every facial triangle $T$.
At this point, we can conclude that no vertex of any facial triangle $T$ lies within $\CH{\Int{P}\cap L_T}$,
but it is still possible that some portion of the interior of a facial triangle $T$ intersects $\CH{\Int{P}\cap L_T}$.
However, as we computed the extreme half-planes for every vertex, we can check in $O(1)$ time whether one of these half-planes is a separating half-plane satisfying \cref{thm:half-plane-cutting-equivalence}(c).

To analyze the runtime, we use an algorithm for intersection detection between half-planes and line segments $O(\log n/\log r)$ times with $O(n)$ half-planes and $O(nr\log n)$ intersections. 
Step 1 and 3 of the above run in time linear in the number of intersections, i.e. $O(nr\log n)$ total time.
\end{proof}

Since line segment intersection can be reduced to ray shooting among halfplanes,
the algorithm of Agarwal and Matou\v{s}ek\cite{AgarwalM93} for ray shooting imply that $T(n, k) = O(n^{3/2+\eps} + n^{1/2+\eps} \cdot k)$ 
for any $\eps > 0$.
Thus we conclude the following corollary by choosing $r= O(1)$.

\begin{cor} 
For any $\eps>0$, there exists a Las Vegas algorithm to determine if a polytope $P$ is carveable by half-planes that runs in time $O(n^{3/2+\varepsilon})$ with high probability. Furthermore, if $P$ is carveable this algorithm outputs a set of cuts to carve $P$.
\end{cor}
We note that using the intersection reporting data structure between triangles\footnote{A half-plane can be simulated by a sufficiently large triangle.} and line segments in $\R^3$ by Ezra and Sharir~\cite{EzraS21} gives a better runtime of $T(n, k) = O(n^{3/2+\eps} + k\log n)$, but does not improve the overall runtime of our algorithm. 
It is plausible to believe that this exponent of ${3/2}$ is the best we can hope for due to lower bounds for Hopcroft's problem in 3D~\cite{Erickson96}.
\section{Ray Sweeps}
\label{sec:ray-cuts}

In this section, we consider 
cutting material with rays.
This models cuts that can be performed with various kinds of tools,
such as a powerful waterjet \footnote{\url{https://youtu.be/pemgwRrCs78}},
or a
laser cutter \footnote{\url{https://youtu.be/J2oyk3ck8Z8}}.
Given a target
polytope $P$, we wish to devise a set of cuts
to carve $P$ out of arbitrary initial
material $C \supset P$. In particular, $C$ is cut into 
pieces, and one of the pieces is $\Int{P}$ (and the remaining are leftovers which can 
be discarded). However, we would like this to be independent of $C$, so we may use the same 
set of cuts to carve $P$ from different pieces of initial material.
We also classify shapes $P$ that admit such a carving.

\begin{figure}
    \centering
    \includegraphics[scale=0.22]{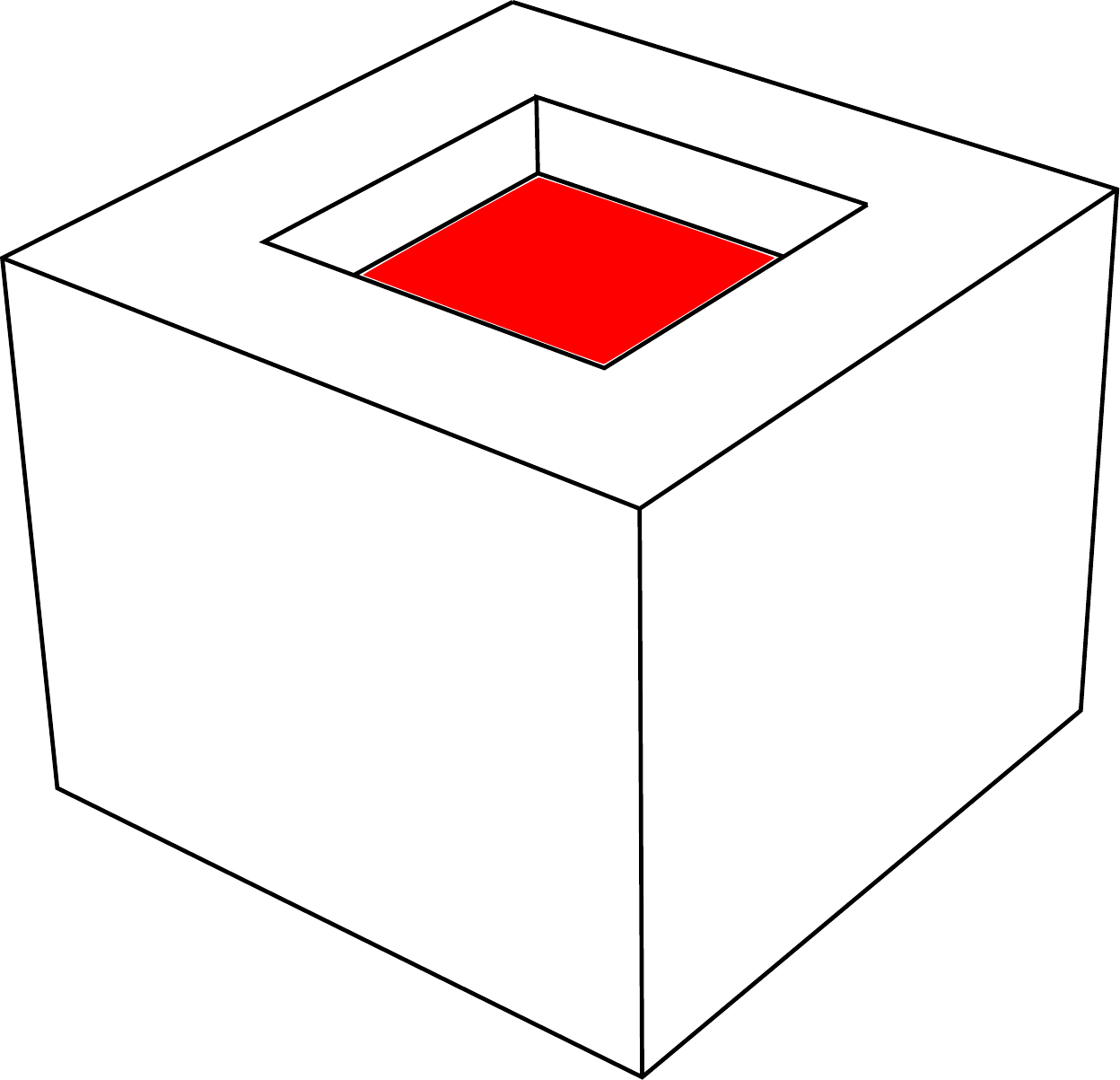}
    \caption{A polytope for which all faces are externally visible that cannot be cut with ray sweeps,
    since ray sweeps require bounded length.
    In particular, the shaded face would require a ray sweep of infinite length (i.e., a space-filling curve).}
    \label{fig:externally-visible-cant-cut}
\end{figure}

\subsection{Model}
We would like our model to capture the finite set of cuts that can be made using
the tools described previously, and they are commonly operated by moving
in a sweeping motion to separate material. Thus, we define our cuts as a set of \textit{ray sweeps}. Consider a ray $R$ defined by an endpoint $a$ and interior point $b$. A sweep is an interpolation where $a$ or $b$ (or both) travel along an arbitrary continuous path of bounded length. 

The reason we require bounded length is because without it, our model would allow us to cut
entire faces using the endpoint of a ray,
via space-filling curves.
\Cref{fig:externally-visible-cant-cut}
is an example of a polytope which \emph{cannot} be cut using ray sweeps,
specifically because of the bounded length requirement.
Without this requirement, any polytope for which all faces are entirely externally visible could be carved.
We use this more restrictive model because it is more representative of how an actual machine could be used
to get flat faces.

Given a (triangulated) polytope $P$, we determine if there is a set of ray sweeps \Rays{} such that, for any $P$-containing set $C\supset P$, the set $C \setminus \del{\cap_{R \in \Rays{}} R }$ has a connected component (i.e., maximal connected open subset) equal to $\Int{P}$.
If such a set of ray sweeps exists, then we also will be able to output it.

As in \cref{subsec:quadratic},
since we only allow finitely many sweeps of bounded length
(and hence do not 
permit 
space-filling curves),
it suffices to ask
if each facial triangle can be cut independently.
Thus,
in this section, we attempt to solve the following
two-dimensional problem:
Given a triangle $T\subset\R^2$ on the plane,
and a set $B_T\subset\R^2$ which is the disjoint union
of simple polygons,
is there a set of ray sweeps \Rays{} inside the plane
such that $T \subset \bigcup_{R\in \Rays{}}R$
and each $R\in \Rays$ has $\Rays\cap\Int{B_T}=\varnothing$?
If we can solve this 2D problem
in time $t(k)$,
where $k$ is the number of vertices forming the polygons of $B_T$,
then we can
classify triangulated polytopes $P$
that can be carved using ray sweeps
in time $O(n\cdot t(n))$.

In fact, we show that sweeps of a special form suffice:
A \emph{linear ray sweep} is one for which the endpoint $a$ can be linearly interpolated between two points,
and the interior point $b$ is constant.

\begin{figure}
    \centering
    \includegraphics[scale=0.75,page=1]{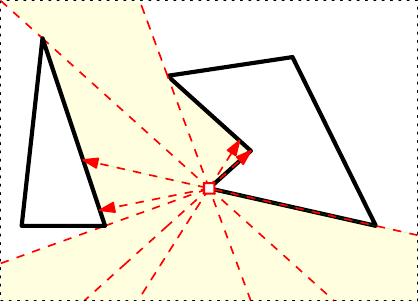}
    \caption{An example of the angle sweep performed in the proof of \cref{thm:find-rays} around a vertex $v$.
        The events are denoted with the red lines/rays, and the valid ray sweeps are coloured.}
    \label{fig:angle-sweep-example}
\end{figure}

\subsection{2D Cuttable Regions}

Compared to \cref{sec:half-plane},
our approach in this section is reversed.
Rather than directly checking if each triangle can be carved,
we first map out the carveable regions along the plane containing
the triangle,
and then afterwards we check if the triangle itself is contained within
those regions.

\begin{theorem}
\label{thm:find-rays}
Given a set $B\subset\R^2$ which is the disjoint union of simple polygons
with a total of $k$ vertices, there is an $O(k^2\log k)$  
time algorithm that can find a set of linear ray sweeps
with total complexity $O(k^2)$, the union of which is exactly
the union of all rays $R \subset \Rays{}$ such that $R \cap \Int{B} = \varnothing$.
\end{theorem}

\begin{proof}
For each vertex $v$ of $B$, we compute a set of linear ray sweeps
$A_v$, each of which passes through $v$,
and none of which
intersect $\Int{B}$.
At a high level, this algorithm
rotates a line passing through $v$,
and maintains the maximal rays in both directions through the line (if any) that
include $v$ and do not intersect $\Int{B}$. 

For the vertex $v$ in $B$, in order to compute $A_v$, we perform an angle sweep
around $v$, where we rotate a line that at all times passes through $v$.
The events of our angle sweep are the set of other vertices in $B$.
Compute the ray starting from $v$ through every other vertex in $B$,
and sort all other vertices in $B$ by the angles of those rays (with an arbitrary branch cut).
The other vertices form the events of our angle sweep.
Now, perform an angle sweep that continuously sweeps a line $l$,
while maintaining a data structure containing
the order of all edge interiors in $B$ (i.e., edges without their endpoints)
intersected crossed by $l$, along with their order relative to $v$ along $l$.
Any standard binary search tree suffices for this purpose.
Note that the wording \emph{crossed} implies that it is okay for $l$ to intersect an edge
interior to which it is parallel.
Between two events,
we check if along $l$ there is a ray from $v$ to infinity
that does not intersect any edge interiors in $B$
(i.e., we check if $v$ is either the first or last element along $l$ according to our data structure).
If so, we have found a linear ray sweep between these two events.
Extend the ray backwards to the first/last element along $l$ (depending on the direction)
at both events.
The linear ray sweep then linearly interpolates between
those two extended rays,
using $v$ as a pivot point. See \cref{fig:angle-sweep-example} for an illustration.

To show the correctness of this algorithm,
first observe that any ray passing through a vertex of $B$
is contained in one of the linear ray sweeps.
Then,
we claim that any valid ray $V^*$ not intersecting a vertex of $B$ can
be cut using rays which intersect at least one vertex of $B$.
To see this,
let $V \subseteq V^*$ be a ray with endpoint $p \in V^*$, an arbitrary point cut by $V^*$. Consider rotating $V$ around $p$ until it hits a vertex (possibly many) of
$B$, call this new ray $V'$. Observe that $V'$ intersects a vertex of $B$ as desired, and still contains $p$.
Since $p$ was an arbitrary point of $V^*$, the claim follows.
Hence, our algorithm produces
linear ray sweeps that together cut all rays not intersecting $B$.

Since there are $k$ vertices and each angle sweep takes $O(k\log k)$ time
both to sort and to perform,
this algorithm runs in $O(k^2\log k)$ time in total.
\end{proof}

\begin{figure}
    \centering
    \includegraphics[scale=0.75,page=2]{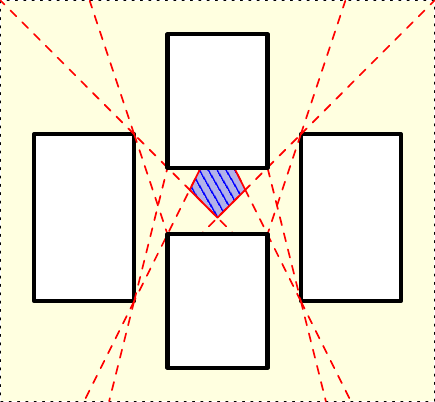}
    \caption{An example of a pentagon that cannot be cut with ray sweeps for a given set of obstacles.
    The bordering linear ray sweeps 
    are shown.}
    \label{fig:angle-sweep-example2}
\end{figure}

\subsection{Decision Algorithm}
\begin{theorem}
\label{thm:raycut-algorithm-2d}
Let $B$ be a set of interior-disjoint simple polygons in $\mathbb{R}^2$
with a total of $k$ vertices,
let $T$ be a triangle in $\mathbb{R}^2$.
Then,
in $O(k^4)$ time,
we can check if there is a set of ray sweeps
$\Rays{}$
such that
$T\subset\cup_{R\in\Rays{}}R$
so that each ray sweep $R\in\Rays{}$ has $R\cap\Int{B}=\varnothing$.
If there is, then we can also output it in the same time complexity.
\end{theorem}

\begin{proof}
Apply
\cref{thm:find-rays}
to get a set of $O(k^2)$
linear ray sweeps whose union is equal to
the set of all area that can be carved via rays that do not intersect $\Int{B}$.
Take the arrangement of all lines containing each ray or line segment
forming the boundary of each ray sweep.
This arrangement can be computed as a doubly-connected edge list in $O(k^4)$ time
using a standard incremental construction~\cite[Theorem 8.6]{BergCKO08}.

For each edge in the graph formed by the arrangement, we include information
about which ray sweep boundaries contain this edge.
Then, we can traverse the cells of the arrangement with a
depth-first search,
while maintaining the current set of regions in which the current cell is contained,
updating as we traverse each edge.
We record, for each cell, whether or not it is part of any region.
In this way, we obtain information about which cells
are included in the union of the regions
in $O(k^4)$ time (linear in the complexity of the arrangement).


Finally, consider each cell of the arrangement.
Check if the cell intersects $T$.
If so, check if it is marked as
being in the interior of at least one linear ray sweep that can be carved.
If it is not, then $T$ cannot be carved.
If we determine that all cells intersecting $T$
can be carved, then $T$ itself can also be carved.
The time complexity to check if each cell intersects $T$
is also $O(k^4)$,
and hence the total time complexity of this algorithm is $O(k^4)$, as desired.
\end{proof}

\begin{cor}
\label{cor:raycut-algorithm-3d}
Let $P$ be a triangulated polytope with $n$ faces.
Then, there is an algorithm running
in $O(n^5)$ time which can determine
if $P$ can be carved with ray sweeps.
Moreover, if it can, then
the algorithm can output a set of linear ray sweeps
carving $P$.
\end{cor}

\begin{proof}
Consider each of the facial triangle $T$ of $P$ separately.
Let $L_T$ be the plane containing $T$,
and compute the set of disjoint open polygonal regions $B_T=L_T\cap\Int{P}$.
If $B_T$ is empty, then $T$ can be cut with a single ray sweep.
If $B_T$ is non-empty, apply 
\cref{thm:raycut-algorithm-2d} to determine if $T$ can be cut with ray sweeps.
Maintain a list of ray sweeps used,
to be output if all triangles can be cut with ray sweeps.
There are $n$ facial triangles,
hence this algorithm runs in $O(n^5)$ time.
\end{proof}

\section{Conclusion}
In this paper we discussed two models (half-planes and ray sweeps)
of carving three-dimensional polytopes.
We focused on the decision variant of each,
while retaining the ability to generate a list 
of cuts when the input polytopes are carveable.
Interestingly, even when a polytope $P$ is not carveable
our algorithms can all be 
modified to 
find a minimal carveable polytope that contains $P$.
This could be quite useful in real-world applications, where a small number of additional cuts could be made using a more specialized tool.

There are several natural resulting open questions from our work:
\begin{itemize}[noitemsep,topsep=0pt,parsep=0pt,partopsep=0pt]
    \item Is there a deterministic algorithm for half-plane carving running in subquadratic time?
    \item Is there a faster algorithm for 2D ray sweep carving that makes use of the triangles to be carved directly?
    \item Are there efficient algorithms for optimization variants of our problems? Either minimizing total length of cuts  
    or minimizing the number of cuts.
\end{itemize}

\small
\bibliographystyle{abbrv}
\bibliography{citations}

\end{document}